\documentclass[11pt]{article}

\usepackage[utf8]{inputenc}
\usepackage[T1]{fontenc}
\usepackage{lmodern}
\usepackage[DIV=11]{typearea} 

\usepackage{microtype}

\usepackage{amsmath}
\usepackage{amssymb}
\usepackage{amsthm}
\usepackage{thmtools}

\usepackage[procnumbered,ruled,vlined,linesnumbered]{algorithm2e}
 
\SetCommentSty{mycommfont} 

\usepackage{xcolor}
\usepackage{xspace}
\usepackage{enumitem}
\usepackage{tikz}

\usepackage{comment}

\usepackage{graphicx}
\usepackage{caption}
\usepackage{subcaption}

\usepackage{authblk}

\def\AdvCite{True} 

\ifdefined\AdvCite

\usepackage[
	style=alphabetic,
	backref=true,
	doi=false,
	url=false,
	maxcitenames=3,
	mincitenames=3,
	maxbibnames=10,
	minbibnames=10,
	backend=bibtex8,
	sortlocale=en_US
]{biblatex}

\usepackage[ocgcolorlinks]{hyperref} 
\usepackage{cleveref}

\DeclareCiteCommand{\citem}
    {}
    {\mkbibbrackets{\bibhyperref{\usebibmacro{postnote}}}}
    {\multicitedelim}
    {}

\renewcommand*{\multicitedelim}{\addcomma\space}

\AtEveryCitekey{\clearfield{note}}

\newcommand{\myhref}[1]{%
	\iffieldundef{doi}
	{\iffieldundef{url}
		{#1}
		{\href{\strfield{url}}{#1}}}
	{\href{http://dx.doi.org/\strfield{doi}}{#1}}%
}

\DeclareFieldFormat{title}{\myhref{\mkbibemph{#1}}}
\DeclareFieldFormat
[article,inbook,incollection,inproceedings,patent,thesis,unpublished]
{title}{\myhref{\mkbibquote{#1\isdot}}}

\makeatletter
\AtBeginDocument{%
	\newlength{\temp@x}%
	\newlength{\temp@y}%
	\newlength{\temp@w}%
	\newlength{\temp@h}%
	\def\my@coords#1#2#3#4{%
		\setlength{\temp@x}{#1}%
		\setlength{\temp@y}{#2}%
		\setlength{\temp@w}{#3}%
		\setlength{\temp@h}{#4}%
		\adjustlengths{}%
		\my@pdfliteral{\strip@pt\temp@x\space\strip@pt\temp@y\space\strip@pt\temp@w\space\strip@pt\temp@h\space re}}%
	\ifpdf
	\typeout{In PDF mode}%
	\def\my@pdfliteral#1{\pdfliteral page{#1}}
	\def\adjustlengths{}%
	\fi
	\ifxetex
	\def\my@pdfliteral #1{}
	\def\adjustlengths{\setlength{\temp@h}{-\temp@h}\addtolength{\temp@y}{1in}\addtolength{\temp@x}{-1in}}%
	\fi%
	\def\Hy@colorlink#1{%
		\begingroup
		\ifHy@ocgcolorlinks
		\def\Hy@ocgcolor{#1}%
		\my@pdfliteral{q}%
		\my@pdfliteral{7 Tr}
		\else
		\HyColor@UseColor#1%
		\fi
	}%
	\def\Hy@endcolorlink{%
		\ifHy@ocgcolorlinks%
		\my@pdfliteral{/OC/OCPrint BDC}%
		\my@coords{0pt}{0pt}{\pdfpagewidth}{\pdfpageheight}%
		\my@pdfliteral{F}
		\my@pdfliteral{EMC/OC/OCView BDC}%
		\begingroup%
		\expandafter\HyColor@UseColor\Hy@ocgcolor%
		\my@coords{0pt}{0pt}{\pdfpagewidth}{\pdfpageheight}%
		\my@pdfliteral{F}
		\endgroup%
		\my@pdfliteral{EMC}%
		\my@pdfliteral{0 Tr}
		\my@pdfliteral{Q}%
		\fi
		\endgroup
	}%
}
\makeatother

\else

\usepackage[ocgcolorlinks]{hyperref} 
\usepackage{cleveref}

\fi

\allowdisplaybreaks[1]

\makeatletter
\g@addto@macro\bfseries{\boldmath}
\makeatother

\makeatletter
\g@addto@macro\mdseries{\unboldmath}
\g@addto@macro\normalfont{\unboldmath}
\g@addto@macro\rmfamily{\unboldmath}
\g@addto@macro\upshape{\unboldmath}
\makeatother

\renewcommand{\paragraph}[1]{\medskip\noindent{\bf #1}\xspace}

\colorlet{DarkRed}{red!50!black}
\colorlet{DarkGreen}{green!50!black}
\colorlet{DarkBlue}{blue!50!black}

\hypersetup{
	linkcolor = DarkRed,
	citecolor = DarkGreen,
	urlcolor = DarkBlue,
	bookmarks = true,
	bookmarksnumbered = true,
	linktocpage = true
}

\declaretheorem[numberwithin=section]{theorem}
\declaretheorem[numberlike=theorem]{lemma}

\declaretheorem[numberlike=theorem]{definition}

\crefname{algorithm}{Algorithm}{Algorithms}
\Crefname{algorithm}{Algorithm}{Algorithms}

\DontPrintSemicolon
\SetKw{KwAnd}{and}
\SetProcNameSty{textsc}
\SetFuncSty{textsc}

\newcommand{\dist}{{\sf dist}}

\newcommand{\diam}{{D}\xspace}

\newcommand{\poly}{\operatorname{poly}}
\newcommand{\polylog}{\operatorname{polylog}}
\newcommand{\source}{\textbf{v}}

\newcommand{\congestion}{{\sf congestion}\xspace}
\newcommand{\dilation}{{\sf dilation}\xspace}

\ifdefined\ShowComment

\def\selfnote#1{{\bf Self Note:} #1}
\def\danupon#1{\marginpar{$\leftarrow$\fbox{D}}\footnote{$\Rightarrow$~{\sf #1 --Danupon}}}
\def\thatchaphol#1{\marginpar{$\leftarrow$\fbox{T}}\footnote{$\Rightarrow$~{\sf #1 --Thatchaphol}}}

\else

\def\selfnote#1{}
\def\danupon#1{}
\def\thatchaphol#1{}

\fi

\title{Distributed Exact Weighted All-Pairs Shortest Paths in $\tilde O(n^{5/4})$ Rounds} 

\author[1]{Chien-Chung Huang}
\author[2]{Danupon Nanongkai}
\author[2]{Thatchaphol Saranurak}
\affil[1]{CNRS, \'Ecole Normale Sup\'erieure, France}
\affil[2]{KTH Royal Institute of Technology, Sweden}

\date{}

\hypersetup{
	pdftitle = {Distributed Exact Weighted All-Pairs Shortest Paths in $\tilde O(n^{5/4})$ Rounds.},
	pdfauthor = {Danupon Nanongkai}
}

\begin{document}

	\begin{titlepage}
		\maketitle
		\pagenumbering{roman}
\begin{abstract}
	
We study computing {\em all-pairs shortest paths} (APSP) on distributed networks (the CONGEST model). The goal is for every node in the (weighted) network to know the distance from every other node using communication.
The problem admits $(1+o(1))$-approximation $\tilde O(n)$-time algorithms ~\cite{LenzenP-podc15,Nanongkai-STOC14},
%
which are matched with $\tilde \Omega(n)$-time lower bounds~\cite{Nanongkai-STOC14,LenzenP_stoc13,FrischknechtHW12}\footnote{$\tilde \Theta$, $\tilde O$ and $\tilde \Omega$ hide polylogarithmic factors. Note that the lower bounds also hold even in the unweighted case and in the weighted case with polynomial approximation ratios.}.  No $\omega(n)$ lower bound or $o(m)$ upper bound were known for exact computation.



In this paper, we present an $\tilde O(n^{5/4})$-time randomized (Las Vegas) algorithm for exact weighted APSP; this provides the first improvement over the naive $O(m)$-time algorithm when the network is not so sparse.
Our result also holds for the case where edge weights are {\em asymmetric} (a.k.a. the directed case where communication is bidirectional). Our techniques also yield an $\tilde O(n^{3/4}k^{1/2}+n)$-time algorithm for the {\em $k$-source shortest paths} problem where we want every node to know distances from $k$ sources; this improves Elkin's recent bound~\cite{Elkin-STOC17} when $k=\tilde \omega(n^{1/4})$.  
	
We achieve the above results by developing distributed algorithms on top of the classic {\em scaling technique}, which we believe is used for the first time for distributed shortest paths computation. One new algorithm which might be of an independent interest is for the {\em reversed $r$-sink shortest paths} problem, where we want every of $r$ sinks to know its distances from all other nodes, given that every node already knows its distance to every sink. We show an $\tilde O(n\sqrt{r})$-time algorithm for this problem. 
Another new algorithm is called {\em short range extension}, where we show that in $\tilde O(n\sqrt{h})$ time the knowledge about distances can be ``extended'' for additional $h$ hops. For this, we use weight rounding to introduce small {\em additive} errors which can be later fixed.

\medskip{\em Remark:} Independently from our result, Elkin recently observed in \cite{Elkin-STOC17} that the same techniques from an earlier version of  the same paper (\url{https://arxiv.org/abs/1703.01939v1}) led to an $O(n^{5/3} \log^{2/3} n)$-time algorithm. 
\end{abstract}

		\newpage
		\setcounter{tocdepth}{2}
		\tableofcontents
	\end{titlepage}

	\newpage
	\pagenumbering{arabic}

	\section{Introduction} \label{sec:intro}

\paragraph{Distributed Graph Algorithms.}
Among fundamental questions in distributed computing is how fast a network can compute its own topological properties, such as minimum spanning tree, shortest paths, minimum cut and maximum flow.  This question has been extensively studied in the so-called {\em CONGEST model} \cite{Peleg00_book} (e.g. \cite{Elkin-STOC17,PanduranganRS-STOC17,HenzingerKN-STOC16,Nanongkai-STOC14,LenzenP_stoc13,DasSarmaHKKNPPW12,Elkin06,PelegR00,GarayKP98}).
In this model (see \Cref{sec:prelim} for details),  a network is modeled by a weighted $n$-node $m$-edge graph $G$. 
Each node represents a processor with unique ID and infinite computational power that initially only knows its adjacent edges and their weights. 
Nodes must communicate with each other in {\em rounds} to discover network properties, where in each round each node can send a message of size $O(\log n)$ to each neighbor. 
The {\em time complexity} is measured as the number of rounds needed to finish the task. It is usually expressed in terms of $n$, $m$, and $\diam$, where $\diam$ is the diameter of the network when edge weights are omitted. Throughout we use $\tilde \Theta$, $\tilde O$ and $\tilde \Omega$ to hide polylogarithmic factors in $n$.

%

%

Note that the whole network can be aggregated to a single node in $O(m)$ time. Thus any graph problem can be trivially solved within $O(m)$ time. A fundamental question is whether this bound can be beaten, and if so, what is the best possible time complexity for solving a particular graph problem. This question has been studied for several decades, marked by a celebrated $O(n\log n)$-time algorithm for the minimum spanning tree (MST) problem by Gallager et al.~\cite{GallagerHS83}. This result was gradually improved and settled with $\tilde \Theta(\sqrt{n} + \diam)$ upper and lower bounds~\cite{PelegR00,GarayKP98,KuttenP98,Awerbuch87,ChinT85,Gafni85}.\footnote{See also \cite{PanduranganRS-STOC17,Elkin-arxiv17-mst} for recent results.}

\paragraph{Approximation vs. Exact Algorithms.} 
Besides MST, almost no other problems were known to admit an $o(m)$-time distributed algorithm when we require the solution to be {\em exact}. 
More than a decade ago, a lot of attention has turned to {\em distributed approximation}, where we allow algorithms to return approximate solutions (e.g. \cite{Elkin04}).  This relaxation has led to a rapid progress in recent years. For example, SSSP, minimum cut, and maximum flow can be $(1+o(1))$-approximated in $\tilde O(\sqrt{n}+\diam)$ time~\cite{HenzingerKN-STOC16,BeckerKKL16,Nanongkai-STOC14,NanongkaiS14_disc,GhaffariK13,GhaffariKKLP15}\footnote{For the maximum flow algorithm, there is an extra $n^{o(1)}$ term in the time complexity.}, and all-pairs shortest paths can be $(1+o(1))$-approximated in $\tilde O(n)$ time~\cite{LenzenP-podc15,Nanongkai-STOC14}; moreover, these bounds are essentially tight up to polylogarithmic factors \cite{DasSarmaHKKNPPW12,Elkin06,PelegR00,KorKP13}. 
%
Given that approximating many graph properties are essentially solved, it is natural to turn back to exact algorithms. A fundamental question is:
\begin{quote}
 {\em Are approximation distributed algorithms more powerful than the exact ones?}
\end{quote}
So far, we only have an answer to the MST problem: due to the lower bound of Das~Sarma~et~al. \cite{DasSarmaHKKNPPW12} (building on \cite{Elkin06,PelegR00,KorKP13}), any $\poly(n)$-approximation algorithm requires $\tilde \Omega(\sqrt{n}+\diam)$ rounds; thus, approximation does {\em not} help. 
For most other problems, however, answering the above question still seems to be beyond current techniques: On the one hand, we are not aware of any lower bound technique that can distinguish distributed $(1+o(1))$-approximation from exact algorithms for the above problems. On the other hand, for most of these problems we do not even know any non-trivial (e.g. $o(m)$-time) exact algorithm. 
%
%
(One exception that we are aware of is SSSP where the classic Bellman-Ford algorithm~\cite{Bellman58,Ford56} takes $O(n)$ time. This bound was recently (in STOC'17) improved by Elkin \cite{Elkin-STOC17}.)



\paragraph{All-Pairs Shortest Paths (APSP).} Motivated by the above question, in this paper we attempt to reduce the gap between the upper and lower bounds for solving APSP exactly.  The goal of the APSP problem is for every node to know the distances from every other node.\footnote{This problem is sometimes referred to as {\em name-independent routing schemes}. See, e.g. \cite{LenzenP_stoc13,LenzenP-podc15}  for discussions and results on another variant called {\em name-dependent routing schemes} which is not considered in this paper.}
Besides being a fundamental problem on its own, this problem is a key component in, e.g., routing tables constructions~\cite{LenzenP_stoc13,LenzenP-podc15}. 

Nanongkai~\cite{Nanongkai-STOC14} and Lenzen and Patt-Shamir~\cite{LenzenP_stoc13,LenzenP-podc15} presented $(1+o(1))$-approximation $\tilde O(n)$-time algorithms, as well as an $\tilde \Omega(n)$ lower bound which holds even when $\diam=O(1)$, when the network is unweighted, and against randomized algorithms.\footnote{In fact, the lower bound holds even for $\poly(n)$-approximation algorithms when the network is weighted and for $(\polylog(n))$-approximation algorithms when the network is unweighted. The same lower bound also holds even for the easier problem of approximating the network diameter~\cite{FrischknechtHW12}. In particular, for the weighted case, the lower bound holds even for $\poly(n)$-approximation algorithms. For the unweighted case, the lower bound holds even for $(3/2-\epsilon)$-approximating the diameter, and even for sparse networks~\cite{AbboudCK16}.}  Very recently, Censor-Hillel~et~al. \cite{Censor-HillelKP17} improved the lower bound to $\Omega(n)$. The same lower bound obviously holds for the exact case.
Neither an $\omega(n)$ lower bound nor an $o(m)$-time algorithm was known, except for some special cases; a notable one is the unweighted case where there are $O(n)$-time algorithms \cite{LenzenP_podc13,HolzerW12,PelegRT12}. 





\paragraph{Results.} Our main result is an $\tilde O(n^{5/4})$-time exact APSP algorithm. Our algorithm is randomized Las Vegas: the output is always correct, and the time guarantee holds both in expectation and with high probability\footnote{We say that an event holds with high probability (w.h.p.) if it holds with probability at least $1-1/n^c$, where $c$ is an arbitrarily large constant.}.
This result provides the first improvement over the naive $O(m)$-time algorithm when the network is not so sparse, and significantly reducing the gap between upper and lower bounds.

Our algorithm also works with the same guarantees when edge weights are {\em asymmetric}, a.k.a. the {\em directed case}. In this case, an edge between nodes $u$ and $v$ can be viewed as two {\em directed} edges, one from $u$ to $v$ and another from $v$ to $u$. These two edges might have different weights, and the weight can be set to infinity. (Note, however, that the infinite weight is not really necessary as it can be replaced by a large $\poly(n)$ weight.) We emphasize that the underlying network is undirected so neither edge direction or weight affect the communication. 
While the previous $(1+o(1))$-approximation $\tilde O(n)$-time algorithms for APSP also work for this case \cite{Nanongkai-STOC14,LenzenP-podc15}, in general it is less understood than the undirected case. For example, while there are tight  $(1+o(1))$-approximation $\tilde O(\sqrt{n}+\diam)$-time algorithms for SSSP~\cite{HenzingerKN-STOC16,BeckerKKL16} in the undirected case, the best known algorithms for the directed case are the $(1+o(1))$-approximation $\tilde O(\sqrt{n\diam}+\diam)$-time one \cite{Nanongkai-STOC14} and the  $\tilde O(\sqrt{n}\diam^{1/4}+\diam)$-time algorithm for the special case called single-source reachability~\cite{GhaffariU15}.

Our techniques also yield an improved algorithm for the {\em $k$-source shortest paths} ($k$-SSP) problem. In this problem, some nodes are marked as source nodes initially (each node knows whether it is marked or not). The goal is for every node to know its distance from every source node. We let $k$ denote the number of source nodes. We show a randomized Las-Vegas $\tilde O(n^{3/4}k^{1/2}+n)$-time algorithm. (Observe that our APSP algorithm is simply a special case when $k=n$.) 
Prior to our work, approximation algorithms and the unweighted case were often considered for this problem (e.g. \cite{ElkinN-FOCS16,HolzerW12,KhanKMPT12,Elkin05-podc01}). 
The only non-trivial exact algorithm known earlier was the algorithm of Elkin~\cite{Elkin-STOC17}.
%
The performance of such algorithm compared to ours is as follows. (We ignore the polylogarithmic terms below for simplicity. For a more precise time guarantee, see \cite{Elkin-STOC17}.)
%
%
\danupon{CHECK: Below is new as Elkin changes his numbers in his final version. Please check.}
\begin{enumerate}[noitemsep]
	\item When $\diam=O(k\sqrt{n})$ and $k = O(\sqrt{n})$, Elkin's algorithm takes $\tilde O(n^{5/6}k^{2/3})$ time. In this case our algorithm is faster when $k=\tilde \omega(n^{1/4})$. 
	\item When $k = \Omega(\sqrt{n})$, Elkin's algorithm takes $\tilde O(n^{2/3}k)$ time. In this case our algorithm is faster when $k=\tilde \omega(n^{1/3})$. 
\end{enumerate}
To conclude, our $k$-SSP algorithm is faster whenever $k=\tilde \omega(n^{1/4})$.

\paragraph{Remarks.}

\noindent 1. The time guarantees of our algorithms depend on the number of bits needed to represent edge weights. This is typically $\polylog(n)$ since edge weights are usually assumed to be positive integers bounded from above by $\poly(n)$; see \Cref{sec:model}. This is the main drawback of the scaling technique that our algorithm heavily relies on (see below). The guarantees of other distributed algorithms we discussed, including Elkin's algorithm \cite{Elkin-STOC17}, do not have this dependency.

\smallskip\noindent 2. Throughout the paper we only show that the output is correct with high probability, but in $O(n)$ time we can check the correctness as follows. First, every node lets its neighbors know about its distances from other nodes (this takes $O(n)$ time). Then, every node checks if it can improve its distance from any node using the distance knowledge from neighbors. If the answer is ``no'' for every node, then the computed distance is correct. If some node answers ``yes'', it can broadcast its answer to all other nodes in $O(n)$ time.

\smallskip\noindent 3. Independently from our result, Elkin recently observed in \cite{Elkin-STOC17} that the same techniques from an earlier version of the same paper (\url{https://arxiv.org/abs/1703.01939v1}) led to an $O(n^{5/3} \log^{2/3} n)$-time algorithm  for APSP on undirected networks.\danupon{TO CONFIRM with Elkin that his algorithm doesn't work for directed case.} 
%
%
Also very recently, Censor-Hillel~et~al.~\cite{Censor-HillelKP17} showed that the standard Alice-Bob framework is incapable of providing a super-linear lower bound for exact weighted APSP, and raised the complexity of APSP as ``an intriguing open question''. They also showed an $\tilde \Omega(n^2)$ lower bound for exactly solving some NP-hard problems, such as minimum vertex cover, maximum independent set and graph coloring. This implies a huge separation between approximation and exact algorithms, since some of these problems can be solved approximately in $O(\polylog(n))$ time.



\subsection{Overview of Algorithms and Techniques.} Our algorithms are built on the {\em scaling technique}. This is a classic technique heavily studied in the sequential setting (e.g. \cite{Goldberg95-soda93,GabowT89,Gabow85-focs83}). As far as we know this is the first time it is used for shortest paths computation in the distributed setting. This technique (see \Cref{sec:scaling} for details) allows us to assume that the distance between any two nodes is $O(n)$ (i.e., the so-called {\em weighted diameter} is $O(n)$). 
%
%
The main challenge here is that edge weight can be {\em zero} (but cannot be negative); without the zero weight, there are already many $\tilde O(n)$-time exact algorithms available (e.g. \cite{LenzenP_podc13,Nanongkai-STOC14}). 
Our algorithms consist of two main subroutines developed for this case, which might be of independent interest. We discuss these subroutines below. To avoid the discussion being too complicated, readers may assume throughout the discussion that the input network has symmetric edge weight. Note however that in reality we have to deal with the asymmetric weights even if the original weight is symmetric. Additionally, we assume for simplicity that every pair of nodes has a unique shortest path.

%

%

\danupon{NEXT ROUND: Easier to say things if define $\dist(u, v)$?}
%

%
\medskip\noindent {\em 1. Short-range extension.} The first subroutine is called {\em short-range-extension}. For simplicity, let us first consider a special case called {\em short-range} problem. In this problem we are given a parameter $h$. The goal is for every node $v$ to know the distance from every node $u$ such that the shortest $uv$-path has at most $h$ edges. 
Previously, this task can be achieved in $\tilde O(nh)$ time by running the Bellman-Ford algorithm for $h$ rounds from every node. By exploiting special properties obtained from the scaling technique, we develop an $\tilde O(n\sqrt{h})$-time algorithm for this problem.  
The main idea is as follows. First we increase the zero weight to a small positive weight $\Delta=1/\sqrt{h}$. By a breadth-first-search (BFS) algorithm, we can solve APSP in the new network (with positive weights) in $\tilde O(n/\Delta)$ time. This solution gives an {\em upper bound} to the APSP problem on the original network (with zero weights). Since we are interested in only shortest paths with at most $h$ edges, it can be argued that the upper bound obtained has an {\em additive error} of $h\Delta$; i.e. it is only $h\Delta$ higher than the actual distance. We fix this additive error by running the Bellman-Ford algorithm for $h\Delta$ rounds from every node. 

The short-range algorithm above can be generalized to the following  {\em short-range-extension} problem. We are given an integer $h$, and initially some nodes in the network already know distances to some other nodes.  For any nodes $u$ and $v$, let $(u=x_0, x_1, x_2, \ldots, x_k=v)$ be the shortest $uv$-path.  We say that $(u, v)$ is {\em $h$-nearly realized} if at least one node among $x_k, x_{k-1}, \ldots, x_{k-h}$ knows its distance from $u$. (Note that the fact that $(u, v)$ is $h$-nearly realized does not necessarily imply that $(v, u)$ is also $h$-nearly realized.) At the end of our algorithm we want to make node $v$ know the distance from $u$, for every nodes $u$ and $v$ such that $(u, v)$ is $h$-nearly realized initially. Observe that the short-range problem is the special case where  initially node $u$ knows distances from no other nodes.
By modifying the short-range algorithm, we can show that this problem can be solved in $\tilde O(n\sqrt{h})$ time as well.

\medskip\noindent {\em 2. Reversed $r$-sink shortest paths.} 
The second subroutine is called {\em reversed $r$-sink shortest paths}. Initially, we assume that every node $v$ knows the distance from $v$ to $r$ sink nodes. The goal is for every sink to know its distance from every node. A naive solution is for every node $v$ to broadcast to the whole network the distance from $v$ to every sink. This takes $O(nr)$ time since there are $O(nr)$ distance information to broadcast. In this paper, we develop an $\tilde O(n\sqrt{r})$-time algorithm for this task. 

The main idea is for every node $v$ to route the distance from $v$ to every sink $t$ through the shortest $vt$-path. If there is a node $x$ that is contained in more than $n\sqrt{r}$  shortest paths (thus there will be too much information going through $x$),  we will call $x$ a {\em bottleneck node}. We can bound the number of bottleneck nodes to  $O(\sqrt{r})$ by a standard argument -- we charge each bottleneck node to $n\sqrt{r}$ distinct shortest paths among $nr$ of them. 
Now, for every shortest $vt$-path that does not contain a bottleneck node, we  route the distance from node $v$ and sink $t$ as originally planned. This takes $\tilde O(n\sqrt{r})$ time since there is $\tilde O(n\sqrt{r})$ bits of information going through each node. 
For shortest $vt$-paths that contain bottleneck nodes, we do the following. For every bottleneck node $c$, we make every node know their distances from and to $c$ by running the Bellman-Ford algorithm starting at $c$. Then every node broadcasts to the whole network its distance to and from every bottleneck node.\thatchaphol{TODO: the technical part did not do exactly this. But this is simpler. We should edit the technical part.} Since there are $\sqrt{r}$ bottleneck nodes, this takes $O(n\sqrt{r})$ time in total. It is not hard to show that every sink $t$ knows the distance from every node $v$ after this step. 
\danupon{SELF NOTE: It doesn't help to recurse with $k$-SSP here because we have to broadcast the distances between bottleneck nodes and every node anyway.}

\medskip\noindent {\em Putting things together.} 
Finally, we sketch how all tools are put together. First we run the short-range algorithm with parameter $h=\sqrt{n}$. Then we sample $\tilde O(\sqrt{n})$ nodes uniformly at random called {\em centers} so that every $h$-hop path contains a center with high probability. Each center $c$ broadcasts to the whole network its distances to some centers that it learns from the short-range algorithms. At this point, every node knows its distance to every center. We invoke the reversed $r$-sink shortest paths algorithm with centers as sink nodes 
(so $r = \tilde O(\sqrt{n})$),
so that every center knows its distance from every node. At this point, it is not hard to prove that every pair of nodes is $h$-nearly realized. So, we finish by invoking the short-range-extension algorithm with parameter $h=\sqrt{n}$. The total time is $\tilde O(n\sqrt{r}+n\sqrt{h})=\tilde O(n^{5/4})$. 

\danupon{DISCUSS: I just realized that we don't need the short-range algorithm at all. In the first step we just run Bellman-Ford from every center. This should make it more clear which steps contribute to the running time. I don't think this changes the running time of any of our algorithms though. But the nice thing about the short range algorithm is that once we understand it, it's easy to see why the short-range-extension algorithm works.}\thatchaphol{I do not follow this comment. I thought Bellman-Ford from $r$ many centers takes $O(nr)$ rounds.}

To extend the above idea to the $k$-source shortest paths problem, we need slight modifications here and there; in particular, (i) we modify the short-range extension and reversed $r$-sink shortest paths algorithms to deal with $k$ source nodes, and (ii) we treat the sampled centers as source nodes since we need to know the distances from and to them. 

%
%
%
%
%
%

%
%


	\section{Preliminaries} \label{sec:prelim}

\subsection{The Model}  \label{sec:model}

In a nutshell, we consider the standard CONGEST model, except that instead of an undirected graph the underlying graph is modeled by a {\em bidirected} graph, i.e. a directed graph in which the reverse of every edge is also an edge. This is because we have to deal with asymmetric edge weight (even when the initial network has symmetric weights).  Additionally, for simplicity we assume that nodes IDs are in the range of $\{0, 1, \ldots, n-1\}$. (This assumption can be achieved in $O(n)$ time.)

More precisely, we model a network by a {\em bidirected} unweighted  $n$-node $m$-edge graph $G$, where nodes model the processors and  edges model the {\em bounded-bandwidth} links between the processors. Let $V(G)$ and $E(G)$ denote the set of nodes and (directed) edges of $G$, respectively. 
The processors  (henceforth, nodes) are assumed to have unique IDs in the range of $\{0, 1, \ldots, n-1\}$ and infinite computational power. (Note again that typically nodes' IDs are assumed to be in the range of $\{1, \ldots, \poly(n)\}$. But in $O(n)$ time the range can be reduced to $\{0, 1, \ldots, n-1\}$.) 
Each node has limited topological knowledge; in particular, it only knows the IDs of its neighbors and knows {\em no} other topological information (e.g., whether its neighbors are linked by an edge or not). Nodes may also accept some additional inputs as specified by the problem at hand.

For the case of graph problems, the additional input is {\em edge weights}.  Let $w:E(G)\rightarrow \{1, 2, \ldots, \poly(n)\}$ be the edge weight assignment.\footnote{Note that it might be natural to include $\infty$ as a possible edge weight. But this is not necessary since it can be replaced by a large weight of value $\poly(n)$.} We refer to network $G$ with weight assignment $w$ as the {\em weighted network}, denoted by $G(w)$. The weight $w(u,v)$ of each edge $(u,v)$ is known only to $u$ and $v$. As commonly done in the literature, 
we will assume that the maximum weight is $\poly(n)$; so, each edge weight can be sent through an edge (link) in one round.
%
We refer to the weight function as {\em symmetric}, or sometimes {\em undirected}, if for every (directed) edge $(u,v)$, $w(u,v)=w(v,u)$. Otherwise, it is called {\em asymmetric}, or sometimes {\em directed}. We note again that the symmetric case is the typical case considered in the literature, but we have to deal with the asymmetric case in our algorithm.

We measure the performance of algorithms by its running time, defined as the worst-case number of {\em rounds} of distributed communication. 
At the beginning of each round, all nodes wake up simultaneously. Each node $u$ then sends an arbitrary message of $O(\log n)$ bits through each edge $(u,v)$, and the message will arrive at node $v$ at the end of the round. 
%
We assume that nodes always know the number of the current round. 
%
%
%
In this paper, the running time is analyzed in terms of the number of nodes ($n$). 
Since $n$ can be computed in $O(\diam)$ time, where $\diam$ is the diameter of $G$, we will assume that every node knows $n$.

%

\subsection{Problems and Notations}\label{sec:notations}

For every nodes $s$ and $t$ in a weighted network $G(w)$, let $\dist_w(s, t)$ be the distance from $s$ to $t$ in $G(w)$. Note that if $w$ is asymmetric then it might be the case that $\dist_w(s, t)\neq\dist_w(t,s)$. 
Let $P^*_{w}(s, t)$ be the shortest path from $s$ to $t$ in $G(w)$; if there are more than one such path, we let $P^*_{w}(s, t)$ be the one with the least number of edges (if there are still more than one, break tie arbitrarily).  We refer to $P^*_{w}(s, t)$ as {\em the} shortest $st$-path.

The goal of the {\em all-pairs shortest paths} (APSP) problem is for every node $t$ to know  $\dist_w(s, t)$ for every node $s$. In the case of {\em $k$-source shortest paths} ($k$-SSP) problem, there is a set $S$ of $k$ {\em source nodes} (every node knows whether it is in $S$ or not). The goal is for  every node $t$ to know  $\dist_w(s, t)$ for every source $s\in S$. When $k=1$, the problem is called {\em single-source shortest paths} (SSSP).

We say that an event holds {\em with high probability} (w.h.p.) if it holds with probability at least $1-1/n^c$, where $c$ is an arbitrarily large constant.


\subsection{Basic Distributed Algorithms}

%
%
%
%
%

\paragraph{The Bellman-Ford Algorithm.} We note the following algorithm for SSSP on network $G(w)$, known as {\em Bellman-Ford} \cite{Bellman58,Ford56}. Let $s$ be the source node. For any node $t$, let $d_w^t(s, t)$ denote the knowledge of $t$ about $\dist_w(s, t)$. Initially, $d_w^t(s, t)=\infty$ for every node $t$, except that $d_w^s(s, s)=0$. The algorithm proceeds as follows. 
	\begin{enumerate}[noitemsep,label=(\roman*)] 
	\item In round 0, every node $t$ sends $d_w^t(s, t)$ to all its neighbors.
	\item When a node $t$ receives the message about $d_w^x(s, x)$ from its neighbors $x$, it uses the new information to decrease the value of $d_w^t(s, t)$. 
	\item If $d_w^t(s, t)$ decreases, then node $t$ sends the new value of $d_w^t(s, t)$ to all its neighbors.
	\item Repeat (ii) and (iii) for $n$ rounds. 
\end{enumerate}

Clearly, the above algorithm takes $O(n)$ rounds. Moreover, it can be proved that when the algorithm terminates $d_w^t(s, t)=\dist_w(s, t)$; i.e. $t$ knows $\dist_w(s,t)$.

\paragraph{Scheduling of Distributed Algorithms.} Consider $k$ distributed algorithms $A_1, A_2  \dots , A_k$. Let \dilation be such that each algorithm $A_i$ finishes in \dilation rounds if it runs individually. Let \congestion be such that there are at most \congestion messages, each of size $O(\log n)$, sent through each edge (counted over all rounds), when we run all algorithms together. We note the following result of Ghaffari \cite{Ghaffari15-scheduling}:

\begin{theorem}[\cite{Ghaffari15-scheduling}]\label{thm:Ghaffari}
There is a distributed algorithm that can execute $A_1, A_2  \dots , A_k$ altogether in $O(\dilation+\congestion\cdot \log n)$ time.
\end{theorem}

\paragraph{Broadcasting.} We need to follow fact following from basic upcasting and downcasting techniques \cite{Peleg00_book}. (The statement is from \cite{LenzenP_podc13}.)

\begin{lemma}\label{lem:broadcast}
Suppose each $v\in V$ holds $k_v \geq 0$ messages of $O(\log n)$ bits each, for a total of $K = \sum_{v\in V} k_v$ messages.
Then all nodes in the network can receive these $K$ messages within $O(K + \diam)$ rounds.
\end{lemma}

%

\subsection{Sampling the Centers}

In the beginning of each iteration, a special node (with ID 0) chooses a subset of \emph{centers} 
uniformly random and broadcasts this information (their IDs) to all other nodes. 
Here we use a lemma of Ullman and Yannakakis~\cite[Lemma~2.2]{UllmanY91}. 

\begin{lemma}[\cite{UllmanY91}]
	If we choose $z$ distinct nodes uniformly at random from an 
$n$-node graph, then the probability that a given (acyclic) path has a sequence of more than 
$(c n \log n)/z$ nodes, none of which is distinguished, is, for sufficiently large $n$, bounded above by 
$2^{-\alpha c}$ for some positive $\alpha$. 
\end{lemma}

The special node chooses $\sqrt{n}\polylog(n)$ 
centers at random and broadcasts this information (the broadcasting can be done in 
$O(\sqrt{n}\polylog(n) + D) = O(n)$ rounds). Then the following lemma 
is a direct consequence of the previous one. 

\begin{lemma} 
\label{lem:centerDistribution}
Let $w$ be any non-negative weight function. For any nodes $s$ and $t$, let $P^*_{w}(s, t)$ be the shortest $st$-path in $G(w)$ as defined in \Cref{sec:notations}. Then, 
with high probability, every $P^*_{w}(s, t)$ can be decomposed into a set 
of subpaths 
$P_0=(s=u_0, \ldots, u_1)$,
$P_1=(u_1, \ldots, u_2)$, $\ldots$, $P_{k-1}=(u_{k-1}, \ldots, u_k=t)$, where
\begin{itemize}[noitemsep]
	\item the $u_i$ are centers for $1 \leq i \leq k-1$.  
         \item each subpath has at most $\sqrt{n}-1$ edges.
\end{itemize}
\end{lemma}

	\section{The Scaling Framework}\label{sec:scaling}

Let $\bar{w}$ denote the given (possibly asymmetric) weight function of the input graph $G$.
We want every node $t$ to know the distances from other nodes $s$ to itself with respect to $\bar{w}$.
We emphasize that every edge $(u,v)$ is directed, i.e., $(u,v)$ is an ordered pair. 
We need the following definitions:
\begin{definition}
	Let $\beta$ be the integer such that $2^{\beta-1}\le\max_{(u,v)\in E(G)} \bar{w}(u,v) <2^{\beta}$.
	For any $0\le i\le\beta$ and edge $(u,v)$, let $w_{i}(u,v)=\left\lfloor \bar{w}(u,v)/2^{\beta-i}\right\rfloor $.
	That is, $w_{i}(u,v)$ is the number represented by the
	first $i$ most significant bits of $\bar{w}(u,v)$ (when we treat the $\beta$-th
	bit as the most significant one). Let $b_{i}(u,v)\in\{0,1\}$ be the
	$i$-th bit in the binary representation of $\bar{w}(u,v)$, i.e., $\bar{w}(u,v)=\sum_{i=0}^{\beta-1}b_{i}(u,v)2^{i}$. 
	\label{def:beta}
\end{definition}
Note that $\beta=O(\log n)$ because the weights of edges in $G$
are polynomial. For any edge $(u,v)$, $w_{0}(u,v)=0$, $w_{\beta}(u,v)=\bar{w}(u,v)$,
and $w_{i+1}(u,v)=2w_{i}(u,v)+b_{\beta-i}(u,v)$ for $0<i<\beta$.
For each $i$, we can treat $w_{i}$ and $b_{i}$ as a weight function.
\begin{definition}
	For any (asymmetric) weight function $\hat{w}$, we denote by $d_{\hat{w}}^{u}(s,t)$
	the knowledge of the node $u$ about $\dist_{\hat{w}}(s,t)$, i.e., the
	distance from $s$ to $t$ with respective the weight $\hat{w}$.
\end{definition}
The algorithm will runs in $\beta$ iterations. At the $i$-th iteration,
we assume that for every node $t$ knows the distances from all other
nodes $s$ to itself with respect to the weight $w_{i-1}$, i.e. $d_{w_{i-1}}^{t}(s,t)=\dist_{w_{i-1}}(s,t)$
for all $s$ and $t$. The goal is to use this information to so that
at the end of the iteration the knowledge of the distances 
with respect to $w_{i}$, i.e. we have $d_{w_{i}}^{t}(s,t)=\dist_{w_{i}}(s,t)$
for all $s$ and $t$. Note that the assumption about the knowledge
holds in the very beginning when $i=1$, because $d_{w_{0}}^{t}(s,t)=\dist_{w_{0}}(s,t)=0$
for all $s$ and $t$ by \Cref{def:beta}.

For convenience, throughout the paper, we fix the iteration $i$.
We denote the weight functions $w:=w_{i}$, $w':=w_{i+1}$ and $b:=b_{\beta-i}$.
That is, we have $w'(u,v)=2w(u,v)+b(u,v)$ for every edge $(u,v)$.
In the beginning, we have $d_{w}^{t}(s,t)=\dist_{w}(s,t)$ and we
want to have $d_{w'}^{t}(s,t)=\dist_{w'}(s,t)$ at the end. 

\subsection{Upper Bounding the Distances}

As $\dist_{w'}(s,t)$ can be a large polynomial for some $s,t$, we can avoid this by working with 
a {\em set} of reduced weights $r_s$ defined as follows. 

\begin{definition} For any node $s$ and edge $e=(u,v)$,  
	let 
	\begin{equation}
	r_s(u,v) = 2 \dist_{w}(s,u)  + w'(u,v) - 2 \dist_{w}(s,v). 
	\end{equation}
	
	\label{def:rs_definition}
\end{definition}

We note that $r_w$ is an asymmetric weight function even if $w$ and $w'$ are symmetric.
The next lemma states some useful properties of $r_s$: 

\begin{lemma} Let $r_s$ be defined as in Definition~\ref{def:rs_definition}. Then the following holds. 
	\label{lem:basicFactRS}
	\begin{enumerate}
		
		\item[(i)] For any edge $e=(u,v)$, $r_s(u,v)\geq 0$. 
		\item[(ii)]  For any nodes $s$ and $t$, $\dist_{r_s}(s, t) \le n-1$.
		\item[(iii)]  For any nodes $s$ and $t$, $\dist_{w'}(s, t)=2\dist_{w}(s, t)+\dist_{r_s}(s, t)$. 
		In fact, any path is a shortest $st$-path in $G(w')$ if and only if it is a shortest $st$-path in $G(r_s)$. 
	\end{enumerate}
\end{lemma}

\begin{proof} For (i), observe that 
	$r_s(u,v) = 2\dist_{w}(s, u)+w'(u,v)-2\dist_{w}(s,v) \geq 
	2\dist_{w}(s, u)+2w(u,v)-2\dist_{w}(s,v) \geq 0$, where 
	the last inequality follows from the triangle inequality. 
	
	For (ii), first notice that 
	
	\begin{equation}
	r_s(P) = w'(P) - 2 \dist_{w}(s,t), \mbox{for any $st$-path $P$.}
	\label{equ:rsp}
	\end{equation}
	
	The above inequality follows easily from definition. 
	Let $P=(s=v_0, v_1,  \ldots, t=v_k)$, for some $k \le n-1$. Then,
	\begin{align*}
	r_{s}(P) & =\sum_{j=0}^{k-1}r_{s}(v_{j},v_{j+1})\\
	& =\sum_{j=0}^{k-1}2\dist_{w}(s,v_{j})+w'(v_{j},v_{j+1})-2\dist_{w}(s,v_{j+1})\\
	& =(\sum_{j=0}^{k-1}w'(v_{j},v_{j+1}))-2\dist_{w}(s,v_{k})\\
	& =w'(P)-2\dist_{w}(s,t)\,.
	\end{align*} 
	
	Now assume that $P$ is a shortest $st$-path in $G(w)$. Then 
	
	\begin{align*}
	\dist_{r_{s}}(s,t)\leq r_{s}(P) & \leq w'(P)-2\dist_{w}(s,t)\\
	& =(\sum_{j=0}^{k-1}2w(v_{j},v_{j+1})+b(v_{j},v_{j+1}))-2\dist_{w}(s,t)\\
	& =(\sum_{j=0}^{k-1}b(v_{j},v_{j+1}))+2w(P)-2\dist_{w}(s,t)\\
	& =\sum_{j=0}^{k-1}b(v_{j},v_{j+1})\leq n-1.
	\end{align*}

	Here the second inequality follows from~(\ref{equ:rsp}), 
	the fifth equality from the assumption
	that $P$ is a shortest path in $G(w)$ and the last inequality from the fact that 
	$k \leq n-1$ and $b(v_j,v_{j+1}) \in \{0,1\}$. This proves (ii).  
	
	Finally for (iii), let $P= (s=v_0, v_1, \ldots, t=v_k)$ be a shortest $st$-path in $G(w')$. 
	Then 
	
	$$\dist_{r_s}(s,t) \leq r_s(P) = w'(P) - 2 \dist_{w}(s,t)= \dist_{w'}(s,t) - 2 \dist_{w}(s,t),$$ 
	where the last equality holds as $P$ is a shortest $st$-path in $G(w')$. On the 
	other hand, let $P'= (s=v_0, v_1, \ldots, t=v_{k'})$ be a shortest path 
	in $G(r_s)$. Then 
	
	$$\dist_{w'}(s,t) \leq w'(P') = r_s(P') + 2\dist_{w}(s,t) = dist_{r_s}(s,t) + 2\dist_{w}(s,t),$$ 
	
	\noindent where the last equality holds because $P'$ is a shortest path in $G(r_s)$. The above 
	two inequalities establish the first part of (iii), while the second part follows 
	from the first part. 
\end{proof} 

Lemma~\ref{lem:basicFactRS}(ii) implies that shortest path tree with a source $s$, based on $r_s$, has depth at most $n$. However, we {\em cannot} construct 
such a tree using the standard BFS starting from $s$ in just $O(n)$ rounds, the difficulty 
being that it can happen that $r_{s}(u,v)=0$ for some edge $(u,v)$. We also note that Lemma~\ref{lem:basicFactRS}(ii)  does {\em not} imply that every edge in the shortest paths has 0/1-weight.

\section{Main Algorithm}\label{sec:main algo}

In this section, we show the main algorithm described in \Cref{alg:main} 
which is the algorithm for one iteration in the scaling framework from \Cref{sec:scaling}.
The setting is that there are three weight functions $w$, $w'$ and $b$ such that, for every edge $(u,v)$ of the input graph $G$, $b(u,v) \in \{0,1\}$ 
and 
\begin{align}
w'(u,v)&=2w(u,v)+b(u,v).  \label{eq:Input Main}
\end{align}
In the beginning, we have $d_{w}^{t}(s,t)=\dist_{w}(s,t)$
and we want that every node $t$ knows $\dist_{w'}(s,t)$ for every node $s$, i.e., $d_{w'}^{t}(s,t)=\dist_{w'}(s,t)$ 
at the end of the algorithm. 

For every pair of nodes $s$ and $t$, recall that $P^*_{w'}(s, t)$ is the shortest $st$-path in $G(w')$; 
if there are more than one shortest $st$-paths in $G(w')$, 
pick the one with the least number of edges (if there are still more than one, break tie arbitrarily).  
Let $C$ be the set of centers decided in Step 1 of \Cref{alg:main}. 
Next, we define an important definition for our algorithm. Let $P^*_{w'}(s,t)|C$ denote the subpath 
of $P^*_{w'}(s, t)$ from the last center in $C \cap P$ to $t$. If there is no center in $P$, let $P^*_{w'}(s,t)|C = P^*_{w'}(s, t)$.
Let $|P^*_{w'}(s,t)|$ be the number of edges in $P^*_{w'}(s,t)$; similarly, 
$|P^*_{w'}(s,t)|C|$ is the number of edges in $P^*_{w'}(s,t)|C$. 

Recall that by Lemma~\ref{lem:basicFactRS}(iii), $\dist_{w'}(s, t)$ 
differs from $\dist_{r_s}(s, t)$ by $2\dist_{w}(s, t)$, which is known 
to $t$. So if every node $t$ knows that the distances w.r.t. $r_s$ from each node $s$, i.e., $d^t_{r_s}(s,t)=\dist_{r_s}(s,t)$,
then each node $t$ can deduce the the distances w.r.t. to $w'$ as well, i.e., $d_{w'}^{t}(s,t)=\dist_{w'}(s,t)$ for all $s$.
\thatchaphol{This should be moved to short-range section}


\begin{algorithm}
	\caption{Main APSP Algorithm (for one iteration in the scaling framework)}\label{alg:main}
	\KwIn{A graph $G$ and the weight functions $w$, $w'$, and $b$ satisfying \Cref{eq:Input Main}. Every node $t$ knows $\dist_{w}(s,t)$ for every node $s$, i.e., $d_{w}^{t}(s,t)=\dist_{w}(s,t)$. Let $h=\sqrt{n}$.}
	\KwOut{Every node $t$ knows $\dist_{w'}(s,t)$ for every node $s$, i.e. $d_{w'}^{t}(s,t)=\dist_{w'}(s,t)$.}
	
	

	
	
	Node 0 randomly samples $\sqrt{n}\textrm{polylog}(n)$ centers (collectively denoted as $C$) and broadcast their 
	IDs to all other nodes. \tcp{This steps takes $O(n)$ rounds.}
	
	Node $t$ sends $\dist_{w}(s,t)$, for all nodes $s$, to its neighbors $x$ in $G$. The neighbor $x$ internally 
	 uses this knowledge to compute $r_{s}(x,t)$, for all nodes $s$, as defined in \Cref{def:rs_definition}. 
	 \tcp{This steps takes $O(n)$ rounds.}
	
	Apply the \textbf{short-range} algorithm (in \Cref{sec:short-range}) so that every node $s$ knows
	$d_{w'}^s(s,t) \geq \dist_{w'}(s,t)$ for all nodes $t$,
	and if $|P^*_{w'}(s,t)|\leq h$, $d_{w'}^s(s,t)= \dist_{w'}(s,t)$. 
	 \tcp{This step takes $\tilde{O}(n^{1.25})$ rounds.}	
	
	All centers $c \in C$ broadcast their knowledge of $d_{w'}^c(c,c')$, for all centers $c' \in C$,
	to all other nodes in the network. Every node $s$ internally uses this knowledge 
	to calculate $d_{w'}^s(s,c)= \dist_{w'}(s,c)$ for all centers $c \in C$.   \tcp{This step takes $\tilde{O}(n)$ rounds}	
	
	Apply the \textbf{reversed $r$-sink shortest paths} algorithm (in \Cref{sec:reverse})  with  nodes in $C$ as sinks so that every center $c \in C$ knows 
	$d_{w'}^c(s,c)= \dist_{w'}(s,c)$ for all nodes $s$. \tcp{This step takes $\tilde{O}(n^{1.25})$ rounds.}	
		
	Apply the \textbf{short-range-extension} algorithm (in \Cref{sec:short-range-extension}) so that every node $t$ knows  
	$d_{w'}^t(s,t) \geq \dist_{w'}(s,t)$ for all nodes $s$,
	and if $|P^*_{w'}(s,t)|C |\leq h$, $d_{w'}^t(s,t)= \dist_{w'}(s,t)$. \tcp{This step takes $\tilde{O}(n^{1.25})$ rounds.}	
							
\end{algorithm}

We first explain the high-level ideas behind our algorithm. In \Cref{alg:main},
Step 1 is for sampling the centers. Step 2 is needed for the execution
of Steps 3 and 6. Note that the implementation details of Steps 3,
5 and 6 will be elaborated in the subsequent sections.

\paragraph{Correctness:}
Let $h=\sqrt{n}$. For any nodes $s$ and $t$, we will argue that,
after executing Steps 3 to 6, every node $t$ knows the distance w.r.t.
$w'$ from $s$ to $t$, i.e., $d_{w'}^{t}(s,t)=\dist_{w'}(s,t)$.
Let $c_{s}$ be the first node in the path $P_{w'}^{*}(s,t)|C$, i.e.
$P^{*}_{w'}(c_{s},t)=P_{w'}^{*}(s,t)|C$. From the definition, if there
is no centers in $P_{w'}^{*}(s,t)$ then $c_{s}=s$ and otherwise
$c_{s}$ is the last center appeared in the path $P_{w'}^{*}(s,t)$
from $s$ to $t$.

We claim that after Step 5, the node $c_{s}$ will know the distance
w.r.t. $w'$ from $s$ to $c_{s}$, i.e., $d_{w'}^{c_{s}}(s,c_{s})=\dist_{w'}(s,c_{s})$.
If $c_{s}=s$, this is trivial. Suppose $c_{s}\neq s$. Consider the
shortest path $P_{w'}^{*}(s,c_{s})$ from $s$ to $c_{s}$. By \Cref{lem:centerDistribution},
we can partition $P_{w'}^{*}(s,c_{s})$ into subpaths, say $P_{0}=(u_{0}:=s,\ldots,u_{1})$,
$P_{1}=(u_{1},\ldots,u_{2})$, $\ldots$, $P_{k-1}=(u_{k-1},\ldots,u_{k}:=c_{s})$
so that each subpath $P_{j}$ has at most $h-1$ edges for $0\leq j\leq k-1$,
and the $u_{j}$'s are centers for $1\leq j\leq k-1$. As subpath
$P_{j}$ has at most $h-1$ edges, the \textbf{short-range} algorithm
guarantees in \Cref{lem:short range conclude} that $u_{j}$ knows $d_{w'}^{u_{j}}(u_{j},u_{j+1})=\dist_{w'}(u_{j},u_{j+1})$
for $0\leq j\leq k-1$ after Step 3 in \Cref{alg:main}.  
In Step 4,
$d_{w'}^{u_{j}}(u_{j},u_{j+1})$, for $1\leq j\leq k-1$, will broadcast
and be known to $s$. Therefore, after Step 4, the node $s$ would
be able to calculate $\dist_{w'}(s,c_{s})$ and so $d_{w'}^{s}(s,c_{s})=\dist_{w'}(s,c_{s})$.
Then, by the guarantee from \Cref{lem:reversed conclude} of the \textbf{reversed $r$-sink shortest paths} algorithm
in Step 5, the knowledge is ``exchanged'' and so $c_{s}$ knows
$\dist_{w'}(s,c_{s})$, i.e. $d_{w'}^{c_{s}}(s,c_{s})=\dist_{w'}(s,c_{s})$. 

By \Cref{lem:centerDistribution}, we also have that $P_{w'}^{*}(c_{s},t) = P_{w'}^{*}(s,t)|C$
has at most $h-1$ edges. As $d_{w'}^{c_{s}}(s,c_{s})=\dist_{w'}(s,c_{s})$,
by the guarantee of the \textbf{short-range-extension} algorithm by \Cref{lem:short range extension conclude}, 
we have after Step
6 the node $t$ knows the distance $\dist_{w'}(s,t)$, i.e. $d_{w'}^{t}(s,t)=\dist_{w'}(s,t)$
and we are done.

\paragraph{Running Time:}
There are $O(|C|)$ messages to be broadcasted in Step 1, and $O(|C|^{2})$
messages in Step 4. By \Cref{lem:broadcast}, this takes $O(|C|^{2}+\diam)=\tilde{O}(n)$
in total. Step 2 easily takes $O(n)$ rounds (by \Cref{thm:Ghaffari}
we have $\congestion=n$ and $\dilation=1$). In the following three
subsections, we will show that Steps 3, 5 and 6 take $\tilde{O}(n^{1.25})$
rounds each. In particular, \Cref{lem:short range conclude,lem:short range extension conclude} state that the short-range algorithm in Step 3 and the short-range-extension algorithm in Step 6 both take $\tilde{O}(n\sqrt{h})$. \Cref{lem:reversed conclude} states that the reversed $r$-sink shortest paths algorithm in Step 5 takes $\tilde{O}(n\sqrt{|C|})$.
In total, the running time in each iteration is $\tilde{O}(n^{1.25})$
rounds.

\begin{theorem} At the end of \Cref{alg:main}, 
	with high probability, for every node $t$, $d_{w'}^t(s,t)= \dist_{w'}(s,t)$ for all nodes $s$. 
	Furthermore, the algorithm takes $\tilde{O}(n^{1.25})$ rounds. 
\end{theorem}

	\subsection{Short-Range Algorithm} \label{sec:short-range}

In this section we show how to implement Step 3 of \Cref{alg:main} so that 
every node $s$ knows $d_{w'}^s(s,t) \geq \dist_{w'}(s,t)$ for all nodes $t$,
and if $|P^*_{w'}(s,t)|\leq h$, $d_{w'}^s(s,t)= \dist_{w'}(s,t)$. 

The main algorithm in this section is precisely described in \Cref{alg:short range}.
However, it yields a slightly different output: after finishing,
every node $t$ knows $d_{w'}^t(s,t) \geq \dist_{w'}(s,t)$ for all nodes $s$,
and if $|P^*_{w'}(s,t)|\leq h$, $d_{w'}^t(s,t)= \dist_{w'}(s,t)$. 
As they are completely symmetric, we can use \Cref{alg:short range} as an algorithm
for Step 3 of \Cref{alg:main} just by switching the direction of every edge in the graph.
The reason for presenting \Cref{alg:short range} that does not give exactly what we want for Step 3 of \Cref{alg:main} is that,
later in \Cref{sec:short-range-extension}, we will extend \Cref{alg:short range} and obtain the short-range-extension algorithm.
This formulation of \Cref{alg:short range} simplifies the modification a lot.
From now on, we will call \Cref{alg:short range} the short-range algorithm as well.

 
Recall that we mentioned earlier that some edges $(u,v)$ may 
have $r_s(u,v)=0$ and this poses difficulty. Our main idea is to deal with a strictly positive weight function $r'_s$, 
defined as $r_s$ rounded up to the next multiple of $\Delta=\sqrt{1/h}$. More precisely, 

\begin{definition} Let $\Delta =\sqrt{1/h} $. For every node $s$ and every edge $(u,v)$, let 

$r'_s(u,v) =\begin{cases}
	\Delta & \mbox{if $r_s(u,v)=0$, and}\\
	\Delta\lceil r_s(u,v)/\Delta\rceil & \mbox{otherwise.}\\
	\end{cases}$
\label{def:rsprime}
\end{definition} 
%
\thatchaphol{TODO: discuss that this is fine even if it is not integer.}

\begin{algorithm}
	\caption{Short-Range Algorithm}\label{alg:short range}
	\KwIn{Every node $t$ knows $\dist_{w}(s,t)$ and $r_s(t,x)$ for all nodes $s$ and all $t$'s neighbors $x$ in $G$.}
	\KwOut{For every pair of nodes $s$ and $t$, node $t$ knows $d_{w'}^t(s,t) \geq \dist_{w'}(s,t)$ and if
	 $|P^*_{w'}(s,t)|\leq h$, $d_{w'}^t(s,t) = \dist_{w'}(s,t)$.}
	
	For every edge $(u,v)$ and all nodes $s$, both $u$ and $v$ internally compute $r'_s(u,v)$ according to Definition~\ref{def:rsprime}. 
	
	For every node $t$, initially set $d_{r'_s}^t(s,t) = \infty$ for all nodes $s\neq t$ and $d_{r'_t}^t(t,t)=0$. 
	
	For every node $s$, compute SSSP tree from $s$ up to depth $n+h\Delta$ in terms of $r'_s$ by implementing 
	the following BFS: each node $t (\neq s)$ 
	updates $d_{r'_s}^t(s,t)$ according to the message $d_{r'_s}^x(s,x)$ it receives from its neighbor $x$. 
	If $d_{r'_s}^t(s,t) \leq n+ h\Delta$, then in round $d_{r'_s}^t(s,t)/\Delta$, the node $t$ sends $d_{r'_s}^t(s,t)$ to all its neighbors in $G$, if $t$ did not send any message in this step yet. 	
	\tcp{Note that we count the number of rounds from 0.} 
	\label{step:BFS-like}
	
	Every node $t$ sets $d_{r_s}^t(s,t)=\lfloor d_{r'_s}^t(s,t)\rfloor$ for all nodes $s$. (Note that $d_{r_t}^t(t,t)=0$.) Run the following algorithm (which is a modification of the Bellman-Ford algorithm) for every node $s$, in parallel:\; 
	\label{step:Bellman-Ford Modified}

        {
	\begin{enumerate}[noitemsep,label=(\roman*)] 
		\item In round 0, every node $t$ sends $d_{r_s}^t(s,t) $ to all its neighbors.\label{step:Bellman-Ford modified initial}
		\item When a node $t$ receives the message about $d_{r_s}^x(s,x)$ from its neighbors $x$, it uses the new information to decrease the value of $d_{r_s}^t(s,t)$ (as an upper estimate of $\dist_{r_s}(s, t)$). Note that  $d_{r_s}^t(s,t)$ is always an integer.
		\item If $d_{r_s}^t(s,t)$ decreases and $d_{r_s}^t(s,t)\geq d_{r'_s}^t(s,t)-h\Delta$, then the node $t$ sends the new value of $d_{r_s}^t(s,t)$ to all its neighbors. \label{step:Bellman-Ford limited communication step}
		\item Repeat (ii) and (iii) for $h$ rounds. 
	\end{enumerate}
	}

	Every node $t$ calculates $d_{w'}^t(s,t) =  2 \dist_{w}(s,t)+  d_{r_s}^t(s,t) $ for all nodes $s$. 
	

\end{algorithm}

\paragraph{Running Time:}
In \Cref{alg:short range}, Steps 1, 2 and 5 takes no time. For a
single source $s$, the BFS in Step 3 has $\dilation=O((n+h\Delta)/\Delta)=O(n/\Delta+h)$
rounds. As in BFS each node sends messages only once and we run the
BFS in parallel from all nodes $s$, we have $\congestion=O(n)$. By
\Cref{thm:Ghaffari}, we have that Step 3 takes $\tilde{O}(\dilation+\congestion)=\tilde{O}(n/\Delta+h+n)=\tilde{O}(n/\Delta)$.
Step 4 is essentially the Bellman-Ford algorithm except
the following modifications: 
\begin{enumerate}[noitemsep]
	\item we start with $d_{r_{s}}^{t}(s,t)=\lfloor d{}_{r'_{s}}^{t}(s,t)\rfloor$ instead
	of $d_{r_{s}}^{t}(s,t)=\infty$, and 
	\item a node $t$ sends its updated value of $d_{r_{s}}^{t}(s,t)$ only
	when $d_{r_{s}}^{t}(s,t)\geq d_{r'_{s}}^{t}(s,t)-h\Delta$ (instead
	of sending it every time $d_{r_{s}}^{t}(s,t)$ is decreased); see
	Step 4.\ref{step:Bellman-Ford limited communication step}. 
\end{enumerate}
We run the modified Bellman-Ford algorithm for every node $s$ in
parallel. This algorithm for a single source node $s$ has 
$\dilation=O(h)$ and $\congestion = O(h\Delta) = O(\sqrt{h})$ since
every node sends a message to its neighbors at most $O(h\Delta)$
times (due to the second modification). By \Cref{thm:Ghaffari}, parallelizing
$n$ such algorithms takes $\tilde{O}(h+n\cdot h\Delta)=\tilde{O}(nh\Delta)$
rounds. Now it can be concluded that \Cref{alg:short range} takes
$\tilde{O}(n/\Delta+nh\Delta)=\tilde{O}(n\sqrt{h})$ rounds. 

\paragraph{Correctness:}
Next, we show the correctness of  \Cref{alg:short range} using the following lemmas.


\begin{lemma} After Step 3 of Algorithm~\ref{alg:short range}, every node $t$ knows 
	$d_{r'_s}^t(s,t) \geq \dist_{r'_s}(s,t)$ for all nodes $s$ and in particular $d_{r'_s}^t(s,t) = \dist_{r'_s}(s,t)$ if $\dist_{r'_s}(s,t) \leq n + h \Delta$. 
	
	\label{lem:dprimeIsGood}
\end{lemma}

\begin{proof} The first part follows from the property of the BFS. For the second part, first notice that $\Delta$ divides $\dist_{r'_s}(s,t)$ 
	for all nodes $s$ and $t$. By a straightforward induction, it can be shown that by round $\frac{\dist_{r'_s}(s,t)}{\Delta}$, $d_{r'_s}^t(s,t) = \dist_{r'_s}(s,t)$, 
	if $ 0 \leq  \dist_{r'_s}(s,t) \leq n + h\Delta$. 
\end{proof}

\begin{lemma} After Step~\ref{step:Bellman-Ford Modified} of \Cref{alg:short range}, every node $t$ knows 
	$d_{r_s}^t(s,t) \geq \dist_{r_s}(s, t)$ for all nodes $s$, furthermore, if $|P^*_{w'}(s,t)|\leq h$, then 
	$d_{r_s}^t(s,t) = \dist_{r_s}(s, t)$; in particular, $d_{r_s}^t(s,t)$ is decreased to $\dist_{r_s}(s, t)$ in round $|P^*_{w'}(s,t)|$ or before. 
	\label{lem:short range correctness}
\end{lemma}

Observe that the correctness of output of the algorithm 
follows from this lemma, since in Step 5, every note $t$ can correctly compute
$d_{w'}^t(s,t)=\dist_{w'}(s, t)$ if $|P^*_{w'}(s,t)|\leq h$ and otherwise $d_{w'}^t(s,t) \ge \dist_{w'}(s, t)$.

The intuition behind the proof is to show that $\dist_{r'_s}(s, t)$ (stored as $d_{r'_s}^t(s,t)$) computed in Step~\ref{step:BFS-like} is not very far from  $\dist_{r_s}(s, t)$; i.e $\dist_{r'_s}(s, t)-\dist_{r_s}(s,t) \leq h\Delta$. Intuitively, this is because $|P^*_{w'}(s,t)|\leq h$, and for each edge $(u,v)$,  $0 \le r'_s(u,v)-r_s(u,v)\leq \Delta$. This allows us to modify the Bellman-Ford algorithm in Step~\ref{step:Bellman-Ford Modified} to allow a node to speak only when $d_{r'_s}^t(s)- d_{r_s}^t(s,t) \leq h\Delta$. 

\begin{proof}[Proof of \Cref{lem:short range correctness}]
	The fact that after Step 4, $d_{r_s}^t(s,t) \geq \dist_{r_s}(s, t)$ follows easily from induction on the number of rounds. 
	We prove the rest by induction on $|P^*_{w'}(s,t)|$. For the base case where $|P^*_{w'}(s,t)|=0$, i.e. $s=t$, the claim trivially holds as we set $d_{r_t}^t(t,t)=0$ in the beginning of Step~\ref{step:Bellman-Ford Modified}. Now consider any pair of $s$ and $t$, and assume that the lemma
	holds for any $t'$ such that $|P^*_{w'}(s,t')|<|P^*_{w'}(s,t)| \leq h$. 
	Let $x$ be the neighbor of $t$ in $P^*_{w'}(s,t)$, i.e. $\dist_{w'}(s, t)=\dist_{w'}(s, x)+w'(x,t)$.
	Note that 
	\begin{align}
	\dist_{r_s}(s, t)&=\dist_{r_s}(s, x)+r_s(x,t)\nonumber\\
	&=d_{r_s}^{x}(s,x)+r_s(x,t),
	\label{eq: short range correctness}
	\end{align}
	where the first equality holds because  $P^*_{w'}(s,t)=P^*_{r_s}(s,t)$ is a shortest path in $G(r_s)$ by Lemma~\ref{lem:basicFactRS}(iii). 
	The second inequality then holds  by the induction hypothesis. We will be done if the following claim holds.
	%
	
	\textbf{Claim}: $x$ sends the message ``$d_{r_s}^x(s,x)=\dist_{r_s}(s, x)$''
	to $t$ in round $|P^*_{w'}(s,x)|+1$ or before that (equivalently, $d_{r_s}^x(s,x)$ is decreased to $\dist_{r_s}(s, x)$	by round 
	$|P^*_{w'}(s,x)| \le h-1$ or before that). 
	
	To see why we will be done, observe that the claim implies that $t$ can update $d_{r_s}^t(s,t)$ to $\dist_{r_s}(s, t)$ using \Cref{eq: short range correctness} 
	in round $|P^*_{w'}(s,t)|$ or before that. Note that $t$ knows $r_s(x,t)$ from the initial knowledge.
	To prove the claim, we just need to show that 
	$d_{r'_s}^x(s,x) - \dist_{r_s}(s,x) \leq (h-1) \Delta$. We have
	\begin{align*}
	\dist_{r'_{s}}(s,x) & \le\dist_{r{}_{s}}(s,x)+|P_{r_{s}}^{*}(s,x)|\Delta & \mbox{by the definition of }r'_{s}\\
	& \le\dist_{r_{s}}(s,x)+|P_{w'}^{*}(s,x)|\Delta & \mbox{by \Cref{lem:basicFactRS}(iii)}\\
	& \le\dist_{r{}_{s}}(s,x)+(h-1)\Delta \\
	& \le n+(h-1)\Delta & \mbox{by \Cref{lem:basicFactRS}(ii)}
	\end{align*}
	By \Cref{lem:dprimeIsGood}, we have $d_{r'{}_{s}}^{x}(s,x)=\dist_{r'{}_{s}}(s,x)$. By
	the second last inequality, we conclude that $d_{r'{}_{s}}^{x}(s,x)-\dist_{r_{s}}(s,x)\leq (h-1)\Delta$.
	This proves the claim and the entire lemma.
\end{proof}

By flipping the direction of edges in the graph, we can conclude the result that is used in the main algorithm:
\begin{lemma} After running \Cref{alg:short range} on a graph where the direction of each edge is flipped,
	every node $s$ knows 
	$d_{r_s}^s(s,t) \geq \dist_{r_s}(s, t)$ for all nodes $t$, furthermore, if $|P^*_{w'}(s,t)|\leq h$, then 
	$d_{r_s}^s(s,t) = \dist_{r_s}(s, t)$. Moreover the algorithm takes $\tilde{O}(n\sqrt{h})$ rounds.
	\label{lem:short range conclude}
\end{lemma}

	\subsection{Short-Range-Extension Algorithm} \label{sec:short-range-extension}

In this section we show how to implement Step 6 of \Cref{alg:main} with the algorithm called 
short-range-extension algorithm.
We are in the setting such that
in the beginning, every center $c$ already knows $d_{w'}^c(s,c) = \dist_{w'}(s,c)$ for all nodes $s$.
By \Cref{lem:centerDistribution}, this implies with high probability that 
for every pair $s$ and $t$, $(s,t)$ is {\em $h$-nearly realized}. Indeed, let $P^*_{w'}(s,t)=(s=x_0,x_1,x_2,\dots,x_k=t)$ be the shortest path from $s$ to $t$ with respect to $w'$. We have that there is a center $c_s \in \{ x_k,x_{k-1},\dots,x_{k-h}\}$ who knows its distance from $s$ to itself with high probability by \Cref{lem:centerDistribution}.
The goal is that, at the end, every node $t$ knows the distance $\dist_{w'}(s,t)$ for all nodes $s$. 
Moreover, it suffices to show that, at the end, every node $t$ knows $d_{w'}^t(s,t) \geq \dist_{w'}(s,t)$ for all nodes $s$,
and if $|P^*_{w'}(s,t)|C|\leq h$, $d_{w'}^t(s,t)= \dist_{w'}(s,t)$.


The short-range-extension algorithm is a minor modification of the short-range algorithm in \Cref{alg:short range}, with 
the same running time and almost identical implementation. But, in this setting, the centers have additional initial knowledge: every center $t$ already knows $\dist_{w'}(s,t)$ and hence $\dist_{r_s}(s,t)$ for all nodes $s$, 
i.e., $d_{r_s}^t(s,t) = \dist_{r_s}(s,t)$. The following changes exploit this knowledge:
\begin{itemize}
	\item For any node $s$, let $G_{s}$ be the graph obtained from $G$ by
	adding imaginary edges into $G$: for every center $t$, there is
	an additional edge $(s,t)$ with weight $\dist_{r_{s}}(s,t)$. We
	call $G_{s}$ the \emph{$s$-augmented} graph. We define the weight
	function $r''_{s}$ for $G_{s}$ in the same way as how we define
	the weight function $r'_{s}$ for $G$. That is, for each original
	edge $(u,v)$ in $G_{s}$, we set $r''_{s}(u,v)=r'_{s}(u,v)$, and,
	for each imaginary edge $(s,t)$ where $t$ is a center, we set 
	\[
	r''_{s}(s,t)=\begin{cases}
	\Delta & \mbox{if }\dist_{r_{s}}(s,t)=0,\mbox{ and}\\
	\Delta\lceil\dist_{r_{s}}(s,t)/\Delta\rceil & \mbox{otherwise}.
	\end{cases}
	\]
	Let $\dist_{r''_{s}}(u,v)$ denote the distance from $u$ to $v$
	with respect to $r''_{s}$ {\em in the $s$-augmented graph $G_{s}$}. 
	\item In Step 2, every pair of nodes $s$ and $t$, initially set $d_{r''_{s}}^{t}(s,t)=\infty$
	and $d_{r''_{s}}^{t}(t,t)=0$, unless $t$ itself is a center. In
	this case, let 
	\[
	d_{r''_{s}}^{t}(s,t)=\begin{cases}
	\Delta & \mbox{if \ensuremath{\dist_{r_{s}}(s,t)=0}, and}\\
	\Delta\lceil\dist_{r_{s}}(s,t)/\Delta\rceil & \mbox{otherwise.}
	\end{cases}
	\]
	This is possible because each center $t$ already knows $\dist_{r_{s}}(s,t)$
	for all nodes $s$. 
	\item In Step 3, for every node $s$, we compute the same SSSP tree w.r.t.
	$r''_{s}$ instead of $r'_{s}$. Observe that, for every node $s$,
	running the BFS with respect to $r''_{s}$ is the same as simulating
	Step 3 of the original short-range algorithm in the $s$-augmented
	graph $G_{s}$.
	\item In the beginning of Step 4, every node $t$ sets $\lfloor d_{r_{s}}^{t}(s,t)=d_{r''_{s}}^{t}(s,t)\rfloor$
	for all nodes $s$, unless $t$ itself is a center. In this case,
	$d_{r_{s}}^{t}(s,t)=\dist_{r_{s}}(s,t)$. Moreover, we run this step
	for $h+1$ rounds instead of $h$ rounds.
\end{itemize}
The running time clearly does not asymptotically change, and so this
algorithm takes $\tilde{O}(n\sqrt{h})$ rounds. The next two lemmas
establish the correctness of the algorithm and they are close parallels
of Lemmas~\ref{lem:dprimeIsGood} and~\ref{lem:short range correctness}.
\begin{lemma}
	After Step 3 of the modified Algorithm~\ref{alg:short range}, every
	node $t$ knows $d_{r''_{s}}^{t}(s,t)\geq\dist_{r''_{s}}(s,t)$ for
	all nodes $s$, and in particular, $d_{r''_{s}}^{t}(s,t)=\dist_{r''_{s}}(s,t)$
	if $\dist_{r''_{s}}(s,t)\leq n+h\Delta$. \label{lem:dprimeIsGood_Extension} \end{lemma}
\begin{proof}
	The proof is identical to \Cref{lem:dprimeIsGood} except
	that $r'{}_{s}$ is replaced by $r''_{s}$.\end{proof}

\begin{lemma}
	After Step~\ref{step:Bellman-Ford Modified} of the modified Algorithm~\ref{alg:short range},
	every node $t$ knows $d_{r_{s}}^{t}(s,t)\geq\dist_{r_{s}}(s,t)$
	for all nodes $s$, furthermore, if $|P_{w'}^{*}(s,t)|C|\leq h$,
	then $d_{r_{s}}^{t}(s,t)=\dist_{r_{s}}(s,t)$; in particular $d_{r_{s}}^{t}(s,t)$
	decreases to $\dist_{r_{s}}(s,t)$ in round $|P_{w'}^{*}(s,t)|C|+1$.
	\label{lem:short range correctness_extension}\end{lemma}
\begin{proof}
	The proof is almost identical to the proof of Lemma~\ref{lem:short range correctness},
	with the difference that we consider the case that $|P_{w'}^{*}(s,t)|C|\leq h$
	and not $|P_{w'}^{*}(s,t)|\leq h$. Similarly, we prove by induction
	on the length of $|P_{w'}^{*}(s,t)|C|$. For the base case where $|P_{w'}^{*}(s,t)|C|=0$,
	we have that $t$ itself is a center. Hence, the node $t$ already
	knows the distance $\dist_{r_{s}}(s,t)$, i.e., $d_{r_{s}}^{t}(s,t)=\dist_{r_{s}}(s,t)$.
	For the inductive step, we only need to show that $x$, who is the
	previous node of $t$ in $P_{w'}^{*}(s,t)|C$, has decreased $d_{r_{s}}^{x}(s,x)$
	down to $\dist_{r_{s}}(s,x)$ in round $|P_{w'}^{*}(s,x)|C|+1\le h$
	or before that. This follows if we can show $d_{r''_{s}}^{x}(s,x)-\dist_{r_{s}}(s,x)\leq h\Delta$.
	Suppose that $c_{s}$ is the first node in $P_{w'}^{*}(s,t)|C$ which
	is the first node in $P_{w'}^{*}(s,x)|C$ as well. We have that
	\begin{align*}
	\dist_{r''_{s}}(s,x) & \le\dist_{r''_{s}}(s,c_{s})+\dist_{r''_{s}}(c_{s},x)\\
	& \le(\dist_{r_{s}}(s,c_{s})+\Delta)+(\dist_{r_{s}}(c_{s},x)+|P_{r_{s}}^{*}(c_{s},x)|\Delta) & \mbox{by the definition of }r''_{s}\\
	& =\dist_{r_{s}}(s,x)+(|P_{r_{s}}^{*}(c_{s},x)|+1)\Delta\\
	& =\dist_{r_{s}}(s,x)+(|P_{w'}^{*}(c_{s},x)|+1)\Delta & \mbox{by \Cref{lem:basicFactRS}(iii)}\\
	& =\dist_{r_{s}}(s,x)+(|P_{w'}^{*}(s,x)|C|+1)\Delta\\
	& \le\dist_{r_{s}}(s,x)+h\Delta\\
	& \le n+h\Delta & \mbox{by \Cref{lem:basicFactRS}(ii)}
	\end{align*}
	By \Cref{lem:dprimeIsGood_Extension}, we have $d_{r''_{s}}^{x}(s,x)=\dist_{r''_{s}}(s,x)$. By the
	second last inequality, we conclude that $d_{r''_{s}}^{x}(s,x)-\dist_{r_{s}}(s,x)\leq h\Delta$.
	And this completes the induction step and the entire proof. 
\end{proof}

Note that the knowledge about  $\dist_{r_s}(s, t)$ implies the knowledge about $\dist_{w'}(s, t)$. So now we can conclude the lemma that is used in the main algorithm:

\begin{lemma} Suppose that every center $c$ already knows $d_{w'}^c(s,c) = \dist_{w'}(s,c)$ for all nodes $s$.
	After running the modified Algorithm~\ref{alg:short range},
	every node $t$ knows  $d_{w'}^t(s,t) \geq \dist_{w'}(s, t)$ for all nodes $s$, furthermore, if $|P^*_{w'}(s,t)|\leq h$ or $|P^*_{w'}(s,t)|C |\leq h$, then $d_{w'}^t(s,t)=\dist_{w'}(s, t)$. Furthermore, the algorithm runs in 
	$\tilde{O}(n\sqrt{h})$ rounds.
	\label{lem:short range extension conclude}
\end{lemma}

%


	\subsection{Reversed $r$-Sink Shortest Paths Algorithm}\label{sec:reverse}  

In this section, we assume that $r$ special sink nodes $\source_1, \dots, \source_r$ are given and every node $s$ 
knows $d_{w'}^t(s,\source_i)=\dist_{w'}(s,\source_i)$ for all sink nodes $\source_i$. 
(Note that these $r$ special sinks correspond to the centers $C$ in Algorithm~\ref{alg:main}.) We present an $\tilde{O}(n\sqrt{r})$-time algorithm so that each sink 
$\source_i$, $1 \leq i \leq r$, acquires the knowledge $d_{w'}^{\source_i}(s,\source_i)=\dist_{w'}(s, \source_i)$ for all nodes $s$ in the end. 
The algorithm is described in \Cref{alg:reverse}. Here, we write the $t$-sink shortest path tree to mean the shortest path tree (w.r.t. $w'$) that has $t$ as the sink.

\begin{algorithm}
	\caption{Reversed $r$-Sink Shortest Paths Algorithm}\label{alg:reverse}
	\KwIn{$r$ sink nodes $\source_1, \cdots, \source_r$. Every node $s$ knows $\dist_{w'}(s,\source_i)$ for all $1 \leq i \leq r$, i.e., $d_{w'}^{s}(s,\source_i)=\dist_{w'}(s, \source_i)$}
	\KwOut{Each sink node $\source_i$ knows $\dist_{w'}(s, \source_i)$ for all nodes $s$, i.e., $d_{w'}^{\source_i}(s,\source_i)=\dist_{w'}(s, \source_i)$}
	
	Every node $s$ sends $\dist_{w'}(s,\source_i)$, for each $1\leq i \leq r$, to all its neighbors. 
	\label{step:before buildtree}
	
	For each $1 \leq i \leq r$ and every node $s$, $s$ uses the information $\dist_{w'}(x,\source_i)$ from all its neighbors $x$ to decide which neighbor $x^*$ is its parent in the 
	$\source_i$-sink shortest path tree. The node $s$ then informs $x^*$ that it is a child of $x^*$ in the $\source_i$-sink shortest path tree. \label{step:buildtree}
	
	Set $B = \emptyset$. \tcp{$B$ is the set of bottleneck nodes.}	
	
	For each $1 \leq i \leq r$ and every node $s$, $s$ waits until it receives the message $\#(i, x_j)$ from all its children $x_j$ in the $\source_i$-sink shortest path tree. If the node $s \not \in B$, 
	let $\#(i, s) = 1+\sum_{j} \#(i, x_j) $; otherwise $\#(i, s) = 0$. The node $s$ sends $\#(i, s)$ to its parent in the $\source_i$-sink shortest path tree. \label{step:counting}
		
	If any node $s \not \in B \cup \{\source_i\}_{i=1}^{r}$ has $\sum_{i=1}^{r} \#(i,s) > \sqrt{k}n$:\; 
	\label{step:bottleneck}

        {
	\begin{enumerate}[noitemsep,label=(\roman*)] 
		\item $s$ broadcasts its intent of becoming a new bottleneck. 
		\item Node 0 chooses one of the candidates (say the one with the smallest ID) as the new bottleneck $b$ and broadcasts its ID to all nodes. Set $B = B \cup \{b\}$. 
		\item Apply the Bellman-Ford algorithm to build the $b$-sink shortest path tree and the $b$-source shortest path tree, so that every node $s$ knows $d_{w'}^s(s,b)=\dist_{w'}(s,b)$ and $d_{w'}^s(b,s)=\dist_{w'}(b,s)$.
		\item Every node $s$ broadcasts $\dist_{w'}(s,b)$ to all nodes (in particular, to all sinks), so that every sink $\source_i$ knows  $d_{w'}^{\source_i}(s,b)=\dist_{w'}(s,b)$ for all nodes $s$.
		\item Go back to Step~\ref{step:counting}.  
	\end{enumerate}
	}
		
	For each $1 \leq i \leq r$ and each node $s$, $\dist_{w'}(s,\source_i)$ is relayed to sink $\source_i$ through the path $P^*_{w'}(s,\source_i)$ in the $\source_i$-sink shortest path tree if $P^*_{w'}(s,\source_i) \cap B = \emptyset$. That is, every node $x\in V\setminus B$ sends $\dist_{w'}(x,\source_i)$ to its parent in the $\source_i$-sink shortest path tree. When a node $v\in V\setminus B$ receives a message $\dist_{w'}(x,\source_i)$, it sends such message to its parent in the $\source_i$-sink shortest path tree. 
	
	\label{step:propagation}
	
	Each sink $\source_i$, for $1 \leq i \leq r$, computes $\dist_{w'}(s, \source_i)$ for all nodes $s$. \label{step:compute}
\end{algorithm}

Now, we explain the idea of \Cref{alg:reverse}. By Steps 1 and 2, for every sink $\source_i$, 
each node $s$ can decide which neighbor $x^*$ is its parent in the $\source_i$-sink shortest path tree:
if $\dist_w'(s,\source_i) = w'(s,x^*)+\dist_w'(x^*,\source_i)$, then $x^*$ is the parent of $s$. Also, every node $s$ knows which neighbors are its children because the children informed $s$ in Step~\ref{step:buildtree}.


The basic idea is to propagate $\dist_{w'}(\source_i, t)$ for all node $t$ upwards to  $\source_i$ in the $\source_i$-sink shortest path tree (as done in Step~\ref{step:bottleneck}) until $\source_i$ 
receives all the informations. However, a brute-force implementation of this idea leads to $O(nr)$ time complexity, since some nodes may need to send out $O(nr)$ messages. 

We overcome this issue by creating a set $B$ of \emph{bottleneck nodes} (or just \emph{bottlenecks} for short), which is empty initially. 
Intuitively, these nodes are the bottlenecks of the above propagation process. We will let them become a sort of ``ad-hoc'' sinks, namely, if $b \in B$, we will let all nodes $s$ know $d_{w'}^s(s,b) = \dist_{w'}(s,b)$. 
Furthermore, for all $1 \leq i \leq r$, the $\source_i$-shortest path trees will be ``pruned'' from these bottlenecks downwards in the following sense.
In Step~\ref{step:counting}, a node $s$, if not a bottleneck in $B$, aggregates the number of its descendants (including $s$ itself) in the $\source_i$-sink shortest path tree, for each $1 \leq i \leq r$, 
and then informs its parent in the same tree. 
On the other hand, if $t$ is a bottleneck, it informs its parent in the $\source_i$-sink shortest path trees, 
for all $1 \leq i \leq r$,  that it has no descendants, i.e., it is a leaf. 
(this can be regarded as our pruning the $\source_i$-sink shortest path trees from the bottlenecks downwards). 

In Step~\ref{step:bottleneck}, if some nodes $t$, which are neither bottlenecks nor the original sinks, have more than $n\sqrt{r}$ descendants, it declares itself as a potential 
candidate to become a new bottleneck. The special node with ID 0 will then decide on a unique node $b$ to be the new bottleneck (so $B = B \cup \{b\}$) and broadcasts 
this decision. Then we build the $b$-sink shortest path tree and $b$-source shortest path tree using the Bellman-Ford algorithm so that all nodes $s$ knows $d_{w'}^s(s,b)=\dist_{w'}(s,b)$ and $d_{w'}^s(b,s)=\dist_{w'}(b,s)$.
Then, all nodes $s$ forward $d_{w'}^s(s,b)$ to the sinks (by broadcasting to the whole network) so that every sink $\source_i$ knows $d_{w'}^{\source_i}(s,b)=\dist_{w'}(s,b)$. This will be useful information for sinks.
The same process 
(Steps~\ref{step:counting} and~\ref{step:bottleneck}) continues until no more bottleneck is created.  

\begin{lemma} The number of bottlenecks is $|B| = O(\sqrt{r})$ and so Steps~\ref{step:counting} and~\ref{step:bottleneck} repeat $O(\sqrt{r})$ times. 

\label{lem:soManyBottleneck}
\end{lemma}

\begin{proof} Observe that originally the total number of nodes in all $\source_i$-sink shortest path trees, for $1\le i \le r$, is $nr$. Each time a new node becomes a 
bottleneck, all its descendants (at least $\Omega(n\sqrt{r})$ of them) are pruned from these trees. Thus, we can create up to at most 
$O(\frac{nr}{n\sqrt{r}}) = O(\sqrt{r})$ bottlenecks and accordingly Steps~\ref{step:counting} and~\ref{step:bottleneck} repeat the same number 
of times.  
\end{proof} 

When there is no more bottleneck to be created, Step~\ref{step:propagation} simply relays the information $\dist_{w'}(s,\source_i)$ to sink $\source_i$ through the $\source_i$-sink shortest path tree, for each $1 \leq i \leq r$, 
as long as 1) $s \in B$ is a bottleneck, or 2) $s$ is not a bottleneck and the path from $s$ to $\source_i$ in the $\source_i$-sink shortest path tree does not contain a bottleneck, i.e. $P^*_{w'}(s,\source_i) \cap B = \emptyset$.
The last step finishes the algorithm. 

\begin{lemma} In Step~\ref{step:compute}, each sink $\source_i$, for $1 \leq i \leq r$, correctly computes $\dist_{w'}(s, \source_i)$ for all nodes $s$, i.e., $d_{w'}^{\source_i}(s,\source_i)=\dist_{w'}(s, \source_i)$.
\label{lem:lastStepIsRight}
\end{lemma}

\begin{proof} Consider the path $P^*_{w'}(s,\source_i)$ from $s$ to $\source_i$ in the $\source_i$-sink shortest path tree. There are two cases. First, if $s \in B$ or $P^*_{w'}(s,\source_i) \cap B = \emptyset$, then, by Step~\ref{step:propagation}, $\dist_{w'}(s,\source_i)$ is relayed to $\source_i$ and we are done.
	Second, if $s$ is not a bottleneck and there is a bottleneck $b$ in $P^*_{w'}(s,\source_i)$, then, by Step~\ref{step:bottleneck}(iv), $\dist_{w'}(s, b)$ is known to $\source_i$; i.e.  $d_{w'}^{\source_i}(s,b)=\dist_{w'}(s,b)$. 	
	Also, by Step~\ref{step:bottleneck}(iii), $d_{w'}^{\source_i}(b,\source_i)=\dist_{w'}(b, \source_i)$ is known to $\source_i$. 
	Therefore, $\source_i$ can use these pieces of information to correctly compute $\dist_{w'}(s, \source_i) = d_{w'}^{\source_i}(s,b)+d_{w'}^{\source_i}(b,\source_i)$. 
\end{proof} 

The lemma above concludes the correctness of \Cref{alg:reverse}. Now we analyze the running time.

\begin{lemma} \Cref{alg:reverse} takes $\tilde{O}(n\sqrt{r})$ rounds.
	\label{lem:reversed rounds}
\end{lemma}
\begin{proof}
	We will use extensively \Cref{thm:Ghaffari} by analyzing \dilation and \congestion in each step.
	In Steps~\ref{step:before buildtree} and \ref{step:buildtree}, each node only sends $r$ messages to its neighbors. 
	So $\dilation = 1$ and $\congestion = r$, and so this takes $\tilde{O}(r)$ rounds.
	In Step~\ref{step:counting}, for every sink $\source_i$, every node $s$ sends a message once along the $\source_i$-sink shortest path tree. As there can be a path of $n$ hops in the tree, $\dilation = n$. Parallelizing  the processes for all sinks $\source_i$ yields $\congestion = r$. So this step takes $\tilde{O}(n+r)=\tilde{O}(n)$.
	
	Now, we analyze Step~\ref{step:bottleneck}. In Step~\ref{step:bottleneck}(i), at most $n$ nodes need to broadcast one message. By \Cref{lem:broadcast}, this takes $O(n+\diam) = O(n)$ rounds. In Step~\ref{step:bottleneck}(ii), only one node broadcast a message and this takes $O(\diam)$ rounds. In Step~\ref{step:bottleneck}(iii), running Bellman-Ford algorithm for finding the $b$-sink shortest path tree, for one node $b$, takes $O(n)$. In Step~\ref{step:bottleneck}(iv), every node broadcasts one messages and this takes $O(n+\diam) = O(n)$ rounds by \Cref{lem:broadcast}. 
	
	By \Cref{lem:soManyBottleneck}, Steps~\ref{step:counting} and \ref{step:bottleneck} repeat $O(\sqrt{r})$ times. In total, this takes $\tilde{O}(n\sqrt{r})$ rounds. Next, in Step~\ref{step:propagation}, the messages are relayed in the shortest path trees, and so $\dilation=n$. Moreover, $\congestion = O(n\sqrt{r})$ because all the nodes $s$ which are descendants of bottlenecks in any tree do not send messages. So this step also takes $O(n\sqrt{r})$ rounds. Therefore, the total number of rounds of the algorithm is  $O(n\sqrt{r})$.
\end{proof}

Finally, we conclude with the lemma that is used in the main algorithms:
\begin{lemma} Every node $s$ knows $\dist_{w'}(s,\source_i)$ for all sinks $\source_i$ where $1 \leq i \leq r$, i.e., $d_{w'}^{s}(s,\source_i)=\dist_{w'}(s, \source_i)$. Then, running \Cref{alg:reverse}, each sink node $\source_i$ knows $\dist_{w'}(s, \source_i)$ for all nodes $s$, i.e., $d_{w'}^{\source_i}(s,\source_i)=\dist_{w'}(s, \source_i)$. Furthermore, \Cref{alg:reverse} takes $\tilde{O}(n\sqrt{r})$ rounds.
	\label{lem:reversed conclude}
\end{lemma}

	\section{$k$-Source Shortest Paths}\label{sec:kSSP}

In this section, we show how to extend the algorithms presented in \Cref{sec:main algo} to solve the $k$-source shortest paths ($k$-SSP) problem. Recall that in this problem we want every node $v$ to know its distance from every of $k$ sources. We let $S$ be the set of sources. Initially, every node knows whether it is a source or not. 

We modify the APSP algorithm as follows. First, we pick $\beta$ sets of random centers\footnote{Recall the $\beta$ is the number of bits needed to represent edge weight (see \Cref{sec:scaling}).}, where each set has size $$\zeta=\min(k, \sqrt{n})\polylog(n).$$ 
Denoted these sets by $C_1, C_2, \ldots, C_\beta$. (Observe that this step can be done in $\tilde O(\sqrt{n}+\diam)$ time since there are only $\tilde O(\sqrt{n})$ centers in total.)
Now we run each iteration of the scaling framework as in \Cref{sec:main algo}, except that in each iteration we only compute shortest paths from only some sources (instead of all nodes). In particular, the set of sources at iteration $i$ is  $S_i=C_{i+1} \cup C_{i+2} \cup ... \cup C_\beta \cup S$. Thus, we can assume that every node knows its distance from all nodes in $S_{i-1}=C_{i} \cup C_{i+1} \cup \ldots \cup C_\beta \cup S$. We will use $C_i$ as a set of random centers in iteration $i$, in the same way we use $C$ in \Cref{alg:main}.\danupon{I assume that we start from iteration $1$ and stop at iteration $\beta$. (Thus, at iteration $0$ all weights are $0$.)}  \Cref{alg:k-source main} describes the new algorithm in details. In this algorithm, we also need to modify the short-range, reversed $r$-sink shortest paths, and short-range-extension so that they can run faster when there are only $q$ sources, where $q=|S_i$|. This is done as in \Cref{alg:k-source main}.

\danupon{NEXT REVISION: If we think of the scaling framework as a recursive algorithm, the description below will be much cleaner. Also, even in the previous algorithm we don't have to assume anything in the beginning. So things are much easier to describe.}

%
%
%
%

\begin{algorithm}
	\caption{$k$-SSP Algorithm (for iteration $i$ in the scaling framework)}\label{alg:k-source main}
	\KwIn{A graph $G$, weight functions $w$, $w'$, and $b$, and set of $k$ sources $S$. Every node knows whether it is a source or not. Every node $t$ knows $\dist_{w}(s,t)$, i.e. $d_{w}^{t}(s,t)=\dist_{w}(s,t)$,  for every node $s\in S_{i-1}=C_{i} \cup \ldots \cup C_\beta \cup S$. Let $h=n/\zeta=\max(n/k, \sqrt{n})$.}
	\KwOut{Every node $t$ knows $\dist_{w'}(s,t)$ for every node $s\in S_{i}=C_{i+1} \cup \ldots \cup C_\beta \cup S$, i.e. $d_{w'}^{t}(s,t)=\dist_{w'}(s,t)$.}
	
	Let $C=C_i$. 
	
	Node $t$ sends $\dist_{w}(s,t)$, for all nodes $s\in S_{i-1}$, to its neighbors $x$ in $G$. The neighbor $x$ internally 
	uses this knowledge to compute $r_{s}(x,t)$, for all nodes $s\in S_{i-1}$, as defined in \Cref{def:rs_definition}. 
	\tcp{This steps takes $O(q)$ rounds.}
	
	Apply the \textbf{$q$-source short-range} algorithm with nodes in $S_i$ as sources and $h$ as above 
	so that every node $s\in S_i$ knows
	$d_{w'}^s(s,t) \geq \dist_{w'}(s,t)$ for all nodes $t$,
	and if $|P^*_{w'}(s,t)|\leq h$, $d_{w'}^s(s,t)= \dist_{w'}(s,t)$. 
	\tcp{This step takes $\tilde O(\sqrt{nqh})$ $=\tilde O(n\cdot \sqrt{(k+\zeta)/\zeta})$ $= \tilde{O}(n+n^{3/4}k^{1/2})$ rounds.}	
	
	All centers $c \in C$ broadcast their knowledge of $d_{w'}^c(c,c')$, for all centers $c' \in C$,
	to  all other nodes in the network. Every node $s\in S_i$ internally uses this knowledge 
	to calculate $d_{w'}^s(s,c)= \dist_{w'}(s,c)$ for all centers $c \in C$.   \tcp{This step takes $\tilde O(\zeta^2)=\tilde{O}(n)$ rounds}	
	
	Apply the \textbf{reversed $q$-source $r$-sink shortest paths} algorithm with nodes in $S_i$ as sources and nodes in $C$ as sinks,
	so that every center $c \in C$ knows 
	$d_{w'}^c(s,c)= \dist_{w'}(s,c)$ for all nodes $s\in S_i$. \tcp{This step takes $\tilde O(n+\sqrt{nqr})$ $=\tilde O(n+\sqrt{n(k+\zeta)\zeta})$ $= \tilde{O}(n+n^{3/4}k^{1/2})$  rounds.}	
	
	Apply the \textbf{$q$-source short-range-extension} algorithm 
	so that every node $t$ knows  
	$d_{w'}^t(s,t) \geq \dist_{w'}(s,t)$ for all nodes $s\in S_i$,
	and if $|P^*_{w'}(s,t)|C |\leq h$, $d_{w'}^t(s,t)= \dist_{w'}(s,t)$. 	\tcp{This step takes $\tilde O(\sqrt{nqh})$ $= \tilde{O}(n+n^{3/4}k^{1/2})$ rounds.}	
	
\end{algorithm}

\danupon{Maybe it's better to change short-range to $h$-range}

\paragraph{$q$-Source Short-Range(-Extension) Algorithms.}  We round up edge weights to multiples of $\Delta$ as done previously. However, we only run the BFS algorithm from $q$ sources (the depth is still $n+h\Delta$). We also run the modification of the Bellman-Ford algorithm with $q$-sources. By the same analysis as in \Cref{sec:short-range,sec:short-range-extension}, the running time of the $q$-source short-range and short-range-extension algorithms becomes $\tilde O(n/\Delta+h+ qh\Delta)$ which is 
$$\tilde O(\sqrt{nqh})$$ 
when we set $\Delta=\sqrt{n/qh}$.

\paragraph{Reversed $q$-Source $r$-Sink Shortest Paths.} The algorithm proceeds as in \Cref{sec:reverse} except that: 
\begin{itemize}
\item Bottleneck nodes are defined to be those that have $g=\sqrt{nqr}$ messages sent through them.
\item In Step~\ref{step:propagation} the messages are relayed in the shortest path trees only from each source node (and not each node).
%
\end{itemize}

Since there are $qr$ source-sink pairs, there are $|B|\leq \lceil qr/g\rceil = O(1+\sqrt{qr/n})$ bottleneck nodes. By following the proof of \Cref{lem:reversed rounds}, the running time of this algorithm becomes 
$$\tilde O(r+|B|n+\sqrt{nqr})=\tilde O(n+\sqrt{nqr}).$$

\paragraph{Total Time of \Cref{alg:k-source main}.} The $q$-source short-range and short-range-extension algorithms take $\tilde O(\sqrt{nqh})$ $=\tilde O(n\cdot \sqrt{(k+\zeta)/\zeta})$ $= \tilde{O}(n+n^{3/4}k^{1/2})$ rounds. The reversed $q$-source $r$-sink shortest paths algorithm takes $\tilde O(n+\sqrt{nqr})$ $=\tilde O(n+\sqrt{n(k+\zeta)\zeta})$ $= \tilde{O}(n+n^{3/4}k^{1/2})$  rounds.	Other steps can be easily seen to take $\tilde O(n)$ rounds. Thus, \Cref{alg:k-source main} takes  $\tilde{O}(n+n^{3/4}k^{1/2})$  rounds in total.

\section{Open Problems}
The main question is whether distributed APSP can be solved in $\tilde O(n)$ time. Both super-linear lower bound or near-linear upper bound will be a major result. Another related problem is SSSP, where there is still a gap between the lower bound of \cite{DasSarmaHKKNPPW12} and upper bound of \cite{Elkin-STOC17}. 
In general, it is very interesting to close the gap between approximation and exact distributed algorithms. We found this question particular interesting for exact maximum matching and minimum cut; these problem admit an $\tilde \Omega(\sqrt{n})$ lower bound while no non-trivial upper bound is known (even an $O(n)$ one). 
Note that the existing $\tilde \Omega(\sqrt{n})$ lower bound for minimum cut does not hold for a natural special case of checking whether the network has small, e.g. $O(1)$, {\em edge connectivity}. Given that small edge connectivity may indicate the network's likeliness to fail, it is interesting to determine their time complexity exactly. Currently there is a big jump from $O(D)$ time for checking edge connectivity of at most two \cite{Thurimella97,PritchardT11} to $\tilde O(\sqrt{n})$ for higher values \cite{NanongkaiS14_disc}.

	\section{Acknowledgement} 
	This project has received funding from the European Research Council (ERC) under the European Union's Horizon 2020 research and innovation programme under grant agreement No 715672. Nanongkai and Saranurak were also partially supported by the Swedish Research Council (Reg. No. 2015-04659.)
	Nanongkai and Saranurak would like to thank Rotem Oshman for comments on the preliminary version of the result.

	\ifdefined\AdvCite
	
	\printbibliography[heading=bibintoc] 
	
	\else
	
	\bibliographystyle{plain}
	\bibliography{references}
	
	\fi


\end{document}